\documentclass[a4paper]{article}
\usepackage{cite}
\usepackage{algorithm}
\usepackage{amsmath,amssymb,amsfonts}
\usepackage{algorithmic}
\usepackage[dvipdfmx]{graphicx}
\usepackage{color}
\usepackage{textcomp}
\usepackage{amsbsy}
  \usepackage{amsmath, amsfonts, mathrsfs, amsthm,here,colortbl, physics, ascmac, url}
\usepackage{latexsym}
\usepackage{amssymb}
\usepackage{amsbsy}

\newtheorem{theorem}{Theorem}
\newtheorem{ass}[theorem]{Assumption}

\newtheorem{lemma}{Lemma}
\newtheorem{prob}{Problem}

\newcommand{\A}{{\cal A}}
\newcommand{\B}{{\cal B}}

\newcommand{\C}{{\cal C}}
\newcommand{\D}{{\cal D}}

\newcommand{\Hm}{{\cal H}}

\definecolor{accessblue}{cmyk}{1, 0.3, 0, 0.2}
\definecolor{greycolor}{cmyk}{0,0,0,.8}
\def\BibTeX{{\rm B\kern-.05em{\sc i\kern-.025em b}\kern-.08em
T\kern-.1667em\lower.7ex\hbox{E}\kern-.125emX}}
\begin{document}

\title{Cyclic Reformulation Based System Identification for Periodically Time-varying Systems}
\author{
Hiroshi Okajima, Yusuke Fujimoto, Hiroshi Oku and Haruto Kondo
}









\maketitle

\section{INTRODUCTION}\label{section1}

System identification has been an active and well-established area of research for many years. The study of system identification has been presented in \cite{id00,id01,id-access1Robust,id0,id02,id03,id1,id2,id3,id4,id5,id6,id7,id8,id9,id10,id11auto,id11tac}. This research considers modeling through system identification targeting linear periodically time-varying systems (LPTV systems). A time-varying system refers to a system in which the internal parameters of the control object depend on time. While the analysis and design theories for linear time-invariant systems have been extensively studied, time-varying systems have become more challenging to analyze and design. Periodic time-varying systems, in which the parameters change periodically, are relatively easy to handle among such time-varying systems. Objects that have periodicity in their internal parameters exist in various fields. Specific research topics include the control of space tether systems, control of periodic disturbances in optical discs, and multi-agent systems for networks with periodic variations \cite{P3}. Additionally, it is possible to handle multi-rate systems as periodically time-varying systems \cite{P2}. The design methods for control systems for periodic time-varying systems are known in \cite{P1} as well. Periodically time-varying quantizers are used in communication bit rates \cite{bit}. The application range of periodically time-varying systems is comprehensive, and the importance of analysis and design for periodically time-varying systems is still recognized today.

Research on system identification for periodically time-varying systems has been conducted, and interesting research results have been obtained. In \cite{id1}, an identification method using periodic inputs was proposed for a periodically time-varying FIR model, and in \cite{id0}, ensemble identification using multiple input-output datasets was presented. In addition, system identification of periodically time-varying autonomous systems using the lifting method, which is a well-known time-invariant reformulation method, is considered in \cite{real}. Frequency domain system identification methods using periodic inputs are presented in such as \cite{id2,id11tac}. System identification of the periodically time-varying systems is one of the challenging problems to solve. Furthermore, research in a similar field has also been conducted, such as the identification of LPV models using periodic scheduling and inputs in such as \cite{id3,id11auto}. 

Motivated by the above research, we propose a cyclic identification algorithm for periodically time-varying plants. We use the cyclic reformulation\cite{cyc1,cyc2,cyc3,cyc4,cyc5}, which is a time-invariant technique, to deal with system identification of periodically time-varying systems. The periodically time-varying system to be identified is represented by a discrete-time linear state-space model by cyclic reformulation. Cycling re-arranges the signal and converts the parameters of the state-space model into a special structure, allowing the periodically time-varying system to be transformed into an equivalent time-invariant system representation with the same input/output behavior. Furthermore, the state-space model is obtained by using the subspace identification method under the cycled input/output signals. The subspace identification method is one of the identification methods for parametric models and is based on a state-space realization. Additionally, the appropriate state coordinate transformation, which is proposed in this paper, is performed on the obtained linear time-invariant state-space model to obtain the parameters of the original linear periodically time-varying system. This method allows for the determination of the parameters of a periodically time-varying system without any special input signal assumptions.

This paper is organized as follows. Section \ref{sec2} defines the state-space representation of the periodically time-varying system, which is the subject of system identification in this study. We also describe the representation method of cyclic reformulation, which is a time-invariant method used to handle periodically time-varying systems as time-invariant systems and provide specific numerical values and application examples. We also describe the properties of the cyclic reformulation. In section 3, we explain the basics of modeling and system identification, as well as the basic theory of the subspace identification method used in this study. In section 4, we propose a system identification algorithm for periodically time-varying systems. In section 5, we present the results of system identification performed according to the procedure shown in section 4 using numerical examples and verify the obtained model.

Notations: $I_n$ denotes an $n\times n$ identity matrix. $O_{n,m}$ denotes an $n\times m$ zero matrix. 
\section{Problem formulation}\label{sec2}
\subsection{Periodically time-varying system}
An $n$-order discrete-time linear periodically time-varying system $P$ is described by 
\begin{eqnarray}
x(k+1)&=&A_k x(k)+B_k (u(k)+w(k))\label{plant0}\\
y(k) &=& C_k x(k) + D_k u(k)+v(k)\label{plant02}
\end{eqnarray}
$u(k)\in {\bf{R}}^{m}$ is the control input at time $k$, $x(k)\in {\bf{R}}^{n}$ is the state, and $y(k)\in {\bf{R}}^{l}$ is the output. Furthermore, the elements of the input $u(k)$ are denoted as $u_1(k),\cdots,u_m(k)$, respectively. The matrices in the state equation are $A_k \in {\bf{R}}^{n \times n}$, $B_k \in {\bf{R}}^{n \times m}$, $C_k \in {\bf{R}}^{l \times n}$, $D_k \in {\bf{R}}^{l \times m}$, and $w(k)\in {\bf{R}}^{m}$ and $v(k)\in  {\bf{R}}^{l}$ are disturbances. The pair $(C_k, A_k)$ is assumed to be observable\cite{cyc2} for any time point. In addition, the pair $(A_k, B_k)$ is assumed to be controllable. Because the plant (\ref{plant0}), (\ref{plant02}) is observable and controllable at any time point, the following conditions hold for any $k$. 
\begin{eqnarray}
{\bf{rank}} \begin{bmatrix}C_k\\C_{k+1} A_{k}\\C_{k+2} A_{k+1} A_{k} \\ \vdots \\ C_{k+n-1} A_{k+n-2} A_{k+n-3} \cdots A_{k}\end{bmatrix} = n \label{observable}
\end{eqnarray}
\begin{eqnarray}
{\bf{rank}} \left[\begin{matrix}B_{k+n-1}\!&\!A_{k+n-1} B_{k+n-2}\!&\!A_{k+n-1} A_{k+n-2} B_{k+n-3}\end{matrix}\right.\nonumber
\end{eqnarray}
\begin{eqnarray}
\left.\begin{matrix}\cdots& A_{k+n-1} \cdots A_{k+1} B_{k}\end{matrix}\right] = n \label{controllable}
\end{eqnarray}
Let $M$ be the period of the periodically time-varying system ($M$-periodic system). $M$ is assumed to be given in this paper. Then, $A_k = A_{k+M}$, $B_k = B_{k+M}$, $C_k = C_{k+M}$ and $D_k = D_{k+M}$ hold. Let $A_0,\cdots,A_{M-1}$, $B_0,\cdots,B_{M-1}$, $C_0,\cdots,C_{M-1}$ and $D_0,\cdots,D_{M-1}$ be given, then $A_k$, $B_k$, $C_k$ and $D_k$ in (\ref{plant0}) can be re-written as follows. 
\begin{eqnarray}
&A_k = A_{k\,{\bf{mod}}\,M}, B_k = B_{k\,{\bf{mod}\,M}},  \nonumber \\
&C_k = C_{k\,{\bf{mod}}\,M}, D_k = D_{k\,{\bf{mod}},M} \label{eq3}
\end{eqnarray}
Therefore, for system identification of periodically time-varying systems, it is sufficient to find $A_0, \cdots, A_{M-1}$, $B_0, \cdots, B_{M-1}$, $C_0, \cdots, C_{M-1}$ and $D_0, \cdots, D_{M-1}$. In addition, it is equivalent to consider the conditions (\ref{observable}) with (\ref{eq3}) and the following condition. 
\begin{eqnarray}
&{\bf{rank}} \begin{bmatrix}C_{k\,{\bf{mod}}\,M}\\C_{(k+1)\,{\bf{mod}}\,M} A_{k\,{\bf{mod}}\,M}\\C_{(k+2)\,{\bf{mod}}\,M} A_{(k+1)\,{\bf{mod}}\,M} A_{k\,{\bf{mod}}\,M} \\ \vdots \\ C_{(k+n-1)\,{\bf{mod}}\,M} A_{(k+n-2)\,{\bf{mod}}\,M} \cdots A_{k\,{\bf{mod}}\,M}\end{bmatrix} \nonumber \\& = n \label{observable2}
\end{eqnarray}
In similar way, the controllability condition is given as follow. 
\begin{eqnarray}
\lefteqn{{\bf{rank}} \left[\begin{matrix}B_{(k+n-1)\,{\bf{mod}}\,M}\!&\!\!A_{(k+n-1)\,{\bf{mod}}\,M} B_{(k+n-2)\,{\bf{mod}}\,M}\end{matrix}\right.}\nonumber\\ &\left.\begin{matrix}\cdots& A_{(k+n-1)\,{\bf{mod}}\,M} \cdots A_{(k+1)\,{\bf{mod}}\,M} B_{k\,{\bf{mod}}\,M}\end{matrix}\right]\nonumber \\ & = n \label{controllable2}
\end{eqnarray}
We can find from (\ref{observable2}) that it is sufficient to consider the case with $k = 0,\cdots, M-1$ for the $M$-periodic system. 
Based on the above periodic time-varying system setup, the following system identification problem is considered.

\begin{prob}\label{prob1}
Consider the case where input data $\{u(k)\}$ is applied to the plant (\ref{plant0}) and (\ref{plant02}) to obtain output data $\{y(k)\}$. 
Based on the input-output data $\{u(k), y(k)\}_{k=0}^{N-1}$, estimate $A_k$, $B_k$, $C_k$ and $D_k (k=0, \cdots, M-1)$ up to state coordinate transformation.
\hfill $\Box$
\end{prob}
We denote the system matrices $\A_k$, $\B_k$, $\C_k$, and $\D_k$ are the solutions of Problem \ref{prob1}. 
Since we are dealing with an $M$-periodic system, from (\ref{plant0}), (\ref{plant02}) and (\ref{eq3}), we need only matrices $\A_k$, $\B_k$, $\C_k$, $\D_k$ for $k=0,\cdots,M-1$ in Problem \ref{prob1}. 

In addition to the observability conditions, we consider the following matrix $F \in {\bf{R}}^{n\times nl}$ for convenience:  
\begin{eqnarray}
F = \begin{bmatrix} F_1 & F_2 & \cdots & F_n\end{bmatrix},
\end{eqnarray}
where $F_j (j=1,\cdots,n)$ is $n\times l$ matrices. 
We can choose appropriate matrices $F$ such that the matrix rank of the following $n\times n$ matrices $X_k$: 
\begin{eqnarray}
\lefteqn{X_k := }\nonumber \\
&F \begin{bmatrix}C_{k\,{\bf{mod}}\,M}\\C_{(k+1)\,{\bf{mod}}\,M} A_{k\,{\bf{mod}}\,M}\\C_{(k+2)\,{\bf{mod}}\,M} A_{(k+1)\,{\bf{mod}}\,M} A_{k\,{\bf{mod}}\,M} \\ \vdots \\ C_{(k+n-1)\,{\bf{mod}}\,M} A_{(k+n-2)\,{\bf{mod}}\,M} \cdots A_{k\,{\bf{mod}}\,M}\end{bmatrix}\label{fk} 
\end{eqnarray}
becomes $n$ for any $k = 0,\cdots M-1$ when rank condition (\ref{observable2}) holds. 
Then, we determine a matrix $\check X$ as follow:
\begin{eqnarray}
\check X = \bf{diag}\{X_0,X_1,\cdots,X_{M-1}\}\label{checkx}
\end{eqnarray}
The matrix rank of $\check X$ is $Mn$ by appropriate choice of $F$. 

We show a simple example about choosing $F$. We consider the case that $l = 1$. The size of each $F_j$ is $n \times 1$. The following components $(F_j)_i$: 
\begin{eqnarray}\label{fji}
(F_j)_i = \begin{cases} 1, i = j\\0, i\neq j\end{cases}, 
\end{eqnarray}
is one obvious choice for $F_j$ that $F$ satisfy the rank condition of $\check X$. Then, we can see that $F$ is given as an identity matrix $I_{n}$. It is not difficult to choose matrices $F$ which satisfy the conditions of matrix rank. 
\subsection{Time invariant system expression using cyclic reformulation} \label{sec22}
The lifting reformulation is the most traditional method for obtaining a time-invariant system from an original periodically time-varying system. The lifting operation consists of packaging the values of a signal over one period in an extended signal. On the other hand, a method used in this paper is a cyclic reformulation \cite{cyc1,cyc2}. 

At first, a cycled input signal is determined based on the input for (\ref{plant0}) as follows:
\begin{eqnarray}
\lefteqn{\check u(0) = \begin{bmatrix}u(0)\\O_{m,1}\\ \vdots\\ O_{m,1}\end{bmatrix}, \check u(1) = \begin{bmatrix}O_{m,1}\\u(1)\\ \vdots\\ O_{m,1}\end{bmatrix}, \cdots, }\label{checku}\\&\check u(\!M\!-\!1\!) = \begin{bmatrix}O_{m,1}\\ \vdots\\O_{m,1}\\ u(\!M\!-\!1\!)\end{bmatrix}, \check u(M) = \begin{bmatrix}u(M)\\O_{m,1}\\ \vdots\\ O_{m,1}\end{bmatrix}, \cdots \nonumber
\end{eqnarray}
The cycled input $\check u(k)\in {\bf{R}}^{Mm}$ is obtained by using the input $u(k)$. $\check u(k)$ has a unique non-zero sub-vector $u(k)$ at each time-point. The sub-vector $u(k)$ cyclically shifts along the column blocks. In the same manner, we can determine $\check v(k)$ and $\check w(k)$ as the cycled disturbances. 

Then, the cyclic reformulation of the $M$-periodic system (\ref{plant0}), (\ref{plant02}) with (\ref{eq3}) is described by
\begin{equation}
     \label{eqcyclic}
     \begin{array}{rcl}
          \check{x}(k+1) & = & \check{A}\check{x}(k) + \check{B}(\check{u}(k)+\check w(k))\\
          \check{y}(k) & = & \check{C}\check{x}(k) + \check{D}\check{u}(k) + \check v(k), 
     \end{array}
\end{equation}
where matrices $\check{A}, \check{B}, \check{C}, \check{D}$ are given as follows: 
\begin{equation}
     \check{A} = \left[\begin{array}{ccccc}
          O_{n,n} & \cdots & \cdots& O_{n,n} & A_{M-1}\\
          A_0 &  O_{n,n} & \cdots &  O_{n,n} &  O_{n,n}\\
           O_{n,n} & A_1 & \ddots & \vdots & \vdots\\
          \vdots & \ddots & \ddots &  O_{n,n} & \vdots\\
           O_{n,n} & \cdots &  O_{n,n} & A_{M-2} &  O_{n,n} 
     \end{array} \right],\label{checka}
\end{equation}
\begin{equation}
     \check{B} = \left[\begin{array}{ccccc}
           O_{n,m} & \cdots & \cdots& O_{n,m} & B_{M-1}\\
          B_0 & O_{n,m} & \cdots & O_{n,m} & O_{n,m}\\
          O_{n,m} & B_1 & \ddots & \vdots & \vdots\\
          \vdots & \ddots & \ddots & O_{n,m} & \vdots\\
          O_{n,m} & \cdots & O_{n,m} & B_{M-2} & O_{n,m} 
     \end{array} \right],\label{checkb}
\end{equation}
\begin{equation}
     \check{C} = \left[\begin{array}{cccc}
          C_0 & O_{l,n} & \cdots &O_{l,n} \\
          O_{l,n}  & C_1 & \ddots & \vdots\\
          \vdots & \ddots & \ddots & O_{l,n}\\
          O_{l,n}  & \cdots & O_{l,n}  & C_{M-1}
     \end{array} \right],
\end{equation}
\begin{equation}
     \check{D} = \left[\begin{array}{cccc}
          D_0 & O_{l,m} & \cdots &  O_{l,m} \\
           O_{l,m}  & D_1 & \ddots & \vdots\\
          \vdots & \ddots & \ddots &  O_{l,m} \\
           O_{l,m}  & \cdots &  O_{l,m}  & D_{M-1}
     \end{array} \right].
\end{equation}
The dimensions of each matrices are given as $\check{A}\in {\bf{R}}^{Mn\times Mn}$, $\check{B}\in {\bf{R}}^{Mn\times Mm}$, $\check{C}\in {\bf{R}}^{Ml\times Mn}$ and $\check{D}\in {\bf{R}}^{Ml\times Mm}$. The structures of $\check{A}$ and $\check{B}$ are named as cyclic matrices. The structures of $\check{C}$ and $\check{D}$ are block diagonal matrices. The dimensions of the state and output are given as $\check{x}(k)\in {\bf{R}}^{Mn}$, $\check{y}(k)\in {\bf{R}}^{Ml}$. 

The initial state $\check{x}(0)$ is given by using $x(0)$ as follows: 
\begin{eqnarray}
\check{x}(0) = \left[\begin{array}{c}
     x(0) \\ O_{n,1} \\ \vdots \\ O_{n,1}\end{array}\right].\label{siki18}
\end{eqnarray}
Then, it is possible to obtain $\check x(1)$ by using (\ref{eqcyclic}), (\ref{siki18}) with $\check u(0)$ and $\check w(0)$ as follows: 
\begin{eqnarray}
     \check{x}(1) = \left[\begin{array}{c}
           O_{n,1} \\
          A_0x(0) + B_0(u(0)+w(0))\\
           O_{n,1}\\
          \vdots\\
           O_{n,1}
     \end{array} \right]. 
\end{eqnarray}
We can find that a sub-vector in $\check x(1)$ exactly corresponds to $x(1)$ by (\ref{plant0}).  
Furthermore, we can obtain the cycled state signal $\check x(k)$ and the cycled output signal $\check y(k)$ by using (\ref{eqcyclic}) and cycled input signal $\check u(k)$ by using step by step calculation. 

Then, the characteristics of Markov parameters for cyclic reformulation are considered. Markov parameters $\check H(i)$ are coefficients of the impulse response and given by using $\check A$, $\check B$, $\check C$ and $\check D$. 
\begin{eqnarray}
\check H(i) = \begin{cases}\check D,&i = 0 \\ \check C \check A^{i-1} \check B, & i=1,2,\cdots \end{cases}\label{markov}
\end{eqnarray}
By the way, given a positive integer $q$, we introduce a matrix $\check{S}_q$ defined as follows: 
\begin{equation}
     \check{S}_q = \left[\begin{array}{ccccc}
          O_{q,q} & I_q & O_{q,q} & \cdots & O_{q,q}\\
          O_{q,q} & O_{q,q} & I_q & \ddots & \vdots\\
          \vdots & \ddots & \ddots & \ddots & O_{q,q}\\
          O_{q,q} & \ddots & \ddots & \ddots & I_q\\
          I_q & O_{q,q} & \cdots & \cdots & O_{q,q} 
     \end{array} \right].   \label{sq}
\end{equation}
Note that the matrix size of $\check S_q$ is $Mq\times Mq$. $\check S_q$ is a regular matrix and its inverse matrix is a cyclic matrix. For any given block diagonal matrix $E \in R^{Mq\times Mq}$ with block sizes $q\times q$, $\check S_q^{-1}E \check S_q$ also becomes a block diagonal matrix. It should be noted that the individual block elements in $\check S_q^{-1}E \check S_q$ are shifted by one element relative to $E$. 

By using the above matrix $\check S_q$, we provide the following important lemmas for Markov parameters $\check H(i)$, which is given in (\ref{markov}), of the cycled system.
\begin{lemma}\label{lemma1}
Consider the following $Ml\times Mm$ matrix.
\begin{eqnarray}
\check S_l^i \check H(i)\label{mar1}
\end{eqnarray}
Then, $\check S_l^i \check H(i)$ is given as a block diagonal matrix with $l\times m$ block elements for any $i (= 0,1,\cdots)$. In addition, the following matrix: 
\begin{eqnarray}
\check S_l^{i-1} \check H(i) \label{mar2}
\end{eqnarray}
can be regarded as a cyclic matrix. 
\end{lemma}

\begin{proof}
Since $\check S_l^0 = I$ holds, $\check S_l^0 \check H(0) = \check D$ is obviously satisfied and is given as a block diagonal matrix. 
From the fact that the product of $\check S_l$ and the cyclic matrix is a block diagonal matrix, the following matrix is given as a block diagonal matrix by a straightforward operation. 
\begin{eqnarray}
\check S_l^{i-1} \check C \check A^{i-1}\label{lemma1proof1}
\end{eqnarray}
Multiplying (\ref{lemma1proof1}) by $\check S_l$ from the left and the cyclic matrix $\check B$ from the right yields (\ref{mar1}) for any $i=1,2,\cdots$. Therefore, (\ref{mar1}) is given as the block diagonal matrix for any $i$. 

Then, it is obvious the term (\ref{mar2}) is a cyclic matrix because (\ref{mar2}) is the product of $\check S_l^{-1}$ and the block diagonal matrix (\ref{mar1}).  
\end{proof}
In the same manner as the Lemma \ref{lemma1}, we can obtain the following lemma. 
\begin{lemma}\label{lemma2}
Consider the following $Ml\times Mm$ matrix.
\begin{eqnarray}
 \check H(i)\check S_m^i
\end{eqnarray}
Then, $ \check H(i)\check S_m^i$ can be regarded as a block diagonal matrix with $l\times m$ block elements for any $i (= 0,1,\cdots)$. In addition, the following matrix: 
\begin{eqnarray}
 \check H(i)\check S_m^{i-1}
\end{eqnarray}
can be regarded as a cyclic matrix. 
\end{lemma}

\begin{proof}
This lemma can be proved using a similar procedure as Lemma \ref{lemma1}. 
\end{proof}
In addition to Lemmas \ref{lemma1} and \ref{lemma2}, the following lemma can be obtained. 
\begin{lemma}\label{lemma3}
Consider the following $Ml\times Mm$ matrix.
\begin{eqnarray}
 \check S_l^i \check H(i+j)\check S_m^{j}
\end{eqnarray}
Then, $\check S_l^i \check H(i+j)\check S_m^j$ can be regarded as a block diagonal matrix with $l\times m$ block elements for any $i, j (= 0,1,\cdots)$. In addition, the following matrix: 
\begin{eqnarray}
 \check S_l^i \check H(i+j)\check S_m^{j-1}
\end{eqnarray}
can be regarded as a cyclic matrix. 
\end{lemma}
Lemmas \ref{lemma1}, \ref{lemma2}, and \ref{lemma3} given here are valuable properties that hold for the cyclic reformulations. The characteristics shown in Lemmas \ref{lemma1} and \ref{lemma2} are essential ideas for identifying the linear periodically time-varying system and are used later. 

Note that Lemmas \ref{lemma1} and \ref{lemma2} are automatically satisfied if we select appropriate $i$ or $j$ in Lemma \ref{lemma3}. Therefore it is sufficient to handle Lemma \ref{lemma3} as a property of the system with cyclic reformulation. 
In addition to the above valuable lemmas, we give a good property of the systems with cyclic reformulation. 
Matrices $\check F_j$ with the size ($Mn\times Ml$) is determined as follows:
\begin{equation}
     \check{F}_j = \left[\begin{array}{cccc}
          F_j & O_{n,l} & \cdots & O_{n,l}\\
          O_{n,l} & \ddots & \ddots & \vdots\\
          \vdots & \ddots  & F_j & O_{n,l}\\
          O_{n,l} & \cdots & O_{n,l} & F_j
     \end{array} \right]\label{fj}
\end{equation}
By using matrices $\check S_l$ in (\ref{sq}) and $\check F_j$ in (\ref{fj}), the matrix $\check X$, which is determined in (\ref{checkx}), can be derived by hte following calculation with cycled system parameters: 
\begin{eqnarray}
\check X = \sum_{j = 1}^n \check F_j \check S_l^{j-1} \check C \check A^{j-1}. \label{checkx2}
\end{eqnarray}
The matrix rank of $\check X$ is $Mn$ if each $X_i$ in (\ref{fk}) is given as regular matrix. The matrix rank of $\check X$ depends on the observability of the system and selection of $F$. 

Then, we provide an example of $\check {X}$ using a concrete case. A $2$nd-order SISO discrete-time linear periodically time-varying system $P$ is given and its period is $M = 3$. The given system is assumed as observable and controllable. The cyclic reformulation of the plant is written as follows. 
\begin{equation}
     \check{A} = \left[\begin{array}{ccc}
          O_{2,2} & O_{2,2} & A_{2}\\
          A_0 &  O_{2,2} & O_{2,2}\\
          O_{2,2} & A_{1} &  O_{2,2} 
     \end{array} \right], \check{B} = \left[\begin{array}{ccc}
          O_{2,1} & O_{2,1} & B_{2}\\
          B_0 &  O_{2,1} & O_{2,1}\\
          O_{2,1} & B_{1} &  O_{2,1} 
     \end{array} \right], \nonumber
\end{equation}
\begin{equation}
     \check{C} = \left[\begin{array}{ccc}
          C_0 & O_{1,2} &O_{1,2} \\
          O_{1,2}  & C_1 & O_{1,2}\\
          O_{1,2} & O_{1,2}  & C_{2}
     \end{array} \right], \check{D} = \left[\begin{array}{ccc}
          D_0 & 0 & 0 \\
          0 & D_1 & 0\\
          0 & 0 & D_2
     \end{array} \right]. \nonumber
\end{equation}
Since $l = 1$ hold for SISO system, $\check S_1$ is given as follow. 
\begin{equation}
     \check S_1 = \left[\begin{array}{ccccc}
          0 & 1 & 0\\
          0 & 0 & 1\\
          1 & 0 & 0 
     \end{array} \right]\label{sqrei}
\end{equation}
By using (\ref{fji}), $F_1$ and $F_2$ is selected as follows. 
\begin{equation}
     \check F_1 = \left[\begin{array}{ccc}
          1 & 0 & 0\\
          0 & 0 & 0\\
          0 & 1 & 0\\
           0 & 0 & 0\\
            0 & 0 & 1\\
             0 & 0 & 0
     \end{array} \right],
     \check F_2 = \left[\begin{array}{ccc}
          0 & 0 & 0\\
          1 & 0 & 0\\
          0 & 0 & 0\\
           0 & 1 & 0\\
            0 & 0 & 0\\
             0 & 0 & 1
     \end{array} \right]\label{checkfexam}
\end{equation}
By applying (\ref{sqrei}), (\ref{checkfexam}) to (\ref{checkx2}), $\check X$ is derived as follow. 
\begin{equation}
     \check{X} = \left[\begin{array}{ccc}
          C_0 & O_{1,2} &O_{1,2} \\
          C_1 A_0 & O_{1,2} &O_{1,2} \\
          O_{1,2}  & C_1 & O_{1,2}\\
          O_{1,2}  & C_2 A_1 & O_{1,2}\\
          O_{1,2} & O_{1,2}  & C_{2}\\
          O_{1,2} & O_{1,2}  & C_{0}A_2
     \end{array} \right]\label{checkxexam}
\end{equation}
It is obvious that (\ref{checkxexam}) collesponds to (\ref{checkx}). 
We can find that the matrix size of $\check X$ in (\ref{checkxexam}) is $6\times 6$ and is regular (full rank) matrix because the pair $(C_k,A_k)$ is observable for any $k$. 
\section{Subspace Identification using Cycled Signals}\label{chapt3}

In this study, we use the subspace identification method since we use the state-space model as the model for the identified parameters. The subspace identification method is a major system identification method based on the state-space realization of dynamical systems. The subspace identification method is constructed from the theory of the realization problem, which identifies the matrices $A$, $B$, $C$, $D$, and the dimension $n$ in the linear time-invariant state-space model from the given input-output data. The advantage of using the subspace identification method is that it can be easily applied to MIMO systems compared to other system identification methods, and it uses numerically stable algorithms such as singular value decomposition and QR decomposition, so the calculation accuracy is high.

Furthermore, the subspace identification method is classified into several categories depending on the weighting matrix used for singular value decomposition. 
We apply the system identification method as follows: At first, we apply input for (\ref{checku}) and obtain an output signal. Then, the cyclic reformulation is applied to the input and output data sets. Moreover, the subspace identification method is applied for the cycled signals and obtains a state space model parameters ($\A_*,\B_*,\C_*,\D_*$). The matrix sizes are $\A_* \in R^{Mn\times Mn}$, $\B_* \in R^{Mn\times Mm}$, $\C_* \in R^{Ml\times Mn}$, $\D_* \in R^{Ml\times Mm}$. 

System identification is performed based on the subspace identification method using the cycled signals $\check u(k)$ and $\check y(k)$. Then, a system obtained by the system identification method is denoted as $\A_*,\B_*,\C_*,\D_*$. The Markov parameters for the obtained system are given as follows: 
\begin{eqnarray}
\check \Hm (i) = \begin{cases}\D_*,&i = 0 \\ \C_*  \A_*^{i-1} \B_*, & i=1,2,\cdots \end{cases}. 
\end{eqnarray}
The following $Ml\times Mm$ matrix is considered in the same manner in Lemma \ref{lemma3}. 
\begin{eqnarray}
\check S_l^{i} \check \Hm (i+j) \check S_m^j, \label{mar12}
\end{eqnarray}

Then, we set up the following assumption related to Lemma \ref{lemma3} for the state space model parameters ($\A_*,\B_*,\C_*,\D_*$). 

\begin{ass}\label{ass1}
The matrix $ \check S_l^{i} \check \Hm (i+j) \check S_m^j$ can be regarded as a block diagonal matrix with $l\times m$ block elements for any $i, j (= 0,1,\cdots)$. $\hfill  \Box$
\end{ass}

In the later section, we verify that Assumption \ref{ass1} is reasonable through a numerical simulation. 

\section{The state coordinate transformation for obtaining cyclic reformulation}

The matrix parameters are obtained as $\A_*, \B_*, \C_*, \D_*$ by using system identification with the cycled signals $\check u(k)$ and $\check y(k)$. Unfortunately, it is expected that $\A_*, \B_*, \C_*, \D_*$ are dense matrices and are not obtained as a cyclic reformulation structure. The state coordinate transformation for the system parameters ($\A_*$, $\B_*$, $\C_*$, $\D_*$) using the specified transformation matrix $T \in {\bf{R}}^{Mn\times Mn}$ is considered for obtaining cyclic reformulation in this section. The transformation matrix will be derived based on Assumption \ref{ass1}.

The state space vector $\check x_* \in {\bf{R}}^{Mn\times 1}$ of a state space model $\A_*, \B_*, \C_*, \D_*$ is determined. 
The transformation matrix $T$ is set to give the matrices as follows. 
\begin{eqnarray}
\check \A = T^{-1}\A_* T, \check \B = T^{-1}\B_*, \check \C = \C_* T, \check \D = \D_* \label{henkan}
\end{eqnarray}
Note that the matrix $T$ must be invertible. New state $\check x_{tf} \in {\bf{R}}^{Mn\times 1}$ is given by $\check x_{tf} = T^{-1} \check x_*$. 
The following state-space model is obtained using the transformation matrix $T$. 
\begin{equation}
     \label{cycsolution}
     \begin{array}{rcl}
          \check{x}_{tf}(k+1) & = & \check{\A}\check{x}_{tf}(k) + \check{\B}\check{u}(k)\\
          \check{y}(k) & = & \check{\C}\check{x}_{tf}(k) + \check{\D}\check{u}(k) 
     \end{array}
\end{equation}

By selecting an appropriate $T$ for the obtained $\A_*, \B_*, \C_*, \D_*$ in system identification, the objective of this study, which is presented in Problem \ref{prob1}, will be achieved if $\check \A, \check \B, \check \C, \check \D$ are given as a cyclic reformulation form. 

Using the matrices obtained as described above, $T^{-1}$ is defined as follows, where $T^{-1}$ is the inverse matrix of the transformation matrix $T$. 
\begin{eqnarray}
T^{-1} = \sum_{j = 1}^n \check F_j \check S_l^{j-1} \C_* \A_{*}^{j-1}\label{coordinate}
\end{eqnarray}
The matrix form of $\check F_j$ is given in (\ref{fj}). 
The following theorem holds for the state coordinate transformation matrix $T$ given as the inverse of (\ref{coordinate}) for the case that $T$ is a regular matrix. Also, we can see that (\ref{coordinate}) and (\ref{checkx2}) are closely related. We should appropriately choose $\check F_j (j= 1,\cdots, n)$ based on the rank condition. 

The following theorem is presented to obtain the cyclic reformulation of the derived model. 
\begin{theorem}
Assuming that the parameters $\A_*, \B_*, \C_*, \D_*$ are obtained via the subspace identification based on the cycled signal and Assumption \ref{ass1} is satisfied for $\A_*, \B_*, \C_*, \D_*$. In addition, the pairs $(\A_*, \B_*)$ and $(\C_*, \A_*)$ are controllable and observable, respectively. Then, the system $\check \A, \check \B, \check \C, \check \D$, which is obtained by the state coordinate transformation (\ref{henkan}) of  $\A_*, \B_*, \C_*, \D_*$ using the transformation matrix $T$ of (\ref{coordinate}), has a structure of the cyclic reformulation. 
\end{theorem}
\begin{proof}
From Assumption \ref{ass1}, $\D_*$ is a block diagonal matrix. We aim to prove that $T^{-1} \A_* T$ and $T^{-1} \B_*$ are cyclic matrices, and that $\C_* T$ is a block diagonal matrix.

If Assumption \ref{ass1} holds, it can be shown that the following matrix exhibits a cyclic structure. 
\begin{eqnarray}
T^{-1}\B_* = \sum_{j = 1}^n \check F_j \check{S}_l^{j-1} \C_* \A_*^{j-1}\B_*\label{proofeq1}
\end{eqnarray}
This is because $\check S_l^{j-1}\C_* \A_*^{j-1} \B_*$ is a cyclic matrix, and $\check F_j$ is a block diagonal matrix for any $j$. Thus, $T^{-1}\B_*$ is a cyclic matrix. 

Then, we prove that $\C_* T$ is given as a block diagonal matrix. 
Following a similar calculation to (\ref{proofeq1}), matrices $T^{-1}\A_*^k\B_* \check S_m^k$ ($k=0,\cdots, n-1$) are regarded as cyclic matrices. In addition, $T^{-1}\A_*^k\B_* \check S_m^{k+1}$ ($k=0,\cdots, n-1$) are given as block diagonal matrices whose matrices sizes are $Mn\times Mm$. We also confirm that $\C_*\A_*^k\B_* S^{k+1}$ ($k=0,\cdots, n-1$) are given as block diagonal matrices based on Assumption \ref{ass1}. 

Then, we determine block diagonal matrices $\check G_j (j=0,\cdots,n-1)$ whose matrix sizes are $Mm \times Mn$. Their block components are given as $G_j$, whose size are $m\times n$, as similar manner as $F_j$. As a simple example of $G_j$ for $m=1$, each $G_j$ has a size of $1 \times n$. The components of $G_j$ are given as follows: 
\begin{eqnarray}
(G_j)_i = \begin{cases} 1, i = j\\0, i\neq j\end{cases}, 
\end{eqnarray}

As the matrices $T^{-1}\A_*^j\B_* \check S_m^{j+1}\check G_j$ are block diagonal matrices for any $j$, the following matrix $Y$ is regarded as a block diagonal matrix whose size is $Mn\times Mn$. 
\begin{eqnarray}\label{proofeq5}
Y = \sum_{j=0}^{n-1}T^{-1}\A_*^j\B_* \check S_m^{j+1}\check G_j
\end{eqnarray}
In addition, $\C_* T Y$ is also a block diagonal matrix due to Assumption \ref{ass1}. In this paper, we assume the periodically time-varying system to be controllable, and $Y$ is considered a regular matrix through the appropriate choice of $\check G_j$. Since $Y$ is a block diagonal matrix, $Y$ is invertible, and $Y^{-1}$ is also a block diagonal matrix. By multiplying the block diagonal matrix $Y^{-1}$ from the right-hand side of the block diagonal matrix $\C_* T Y$, it is clear that $\C_* T$ can be represented as a block diagonal matrix.

Finally, we prove that $T^{-1}\A_* T$ is given as a cyclic matrix. 
The following matrix $Z_{ij}$ is given as a cyclic matrix for any $i,j$ from Assumption \ref{ass1}.
\begin{eqnarray}
Z_{ij} = \check S_l^{i}\C_*\A_*^{i+j} \B_* \check S^{j}
\end{eqnarray}
In addition, $Z_{ij}$ can be rewritten as follows. 
\begin{eqnarray}\label{proofeq2}
Z_{ij} = \check S_l^{i-1}\C_*\A_*^{i-1}T\, T^{-1}\A_* T\, T^{-1} \A_*^{j} \B_* \check S_m^{j+1}
\end{eqnarray}
By using the block diagonal matrices $\check F_i$ and $\check G_j$, it is obvious that the following terms are given as cyclic matrices.
\begin{eqnarray}\label{proofeq3}
\check F_i Z_{ij} \check G_j
\end{eqnarray}
Then, the following matrix $Z$ is obviously given as a cyclic matrix by using the characteristics about (\ref{proofeq2}) and (\ref{proofeq3}).
\begin{eqnarray}\label{proofeq4}
Z &= \left(\sum_{i=1}^{n} \check F_i\check S_l^{i-1}\C_*\A_*^{i-1}T\right) T^{-1}\A_* T\nonumber \\&\cdot\left(\sum_{j=0}^{n-1} T^{-1} \A_*^{j} \B_* \check S_m^{j+1}\check G_j\right)
\end{eqnarray}
In (\ref{proofeq3}), the left hand side term $\sum_{i=1}^{n} \check F_i\check S_l^{i-1}\C_*\A_*^{i-1}T$ is an identity matrix by (\ref{coordinate}). The right hand side term is $Y$, which is determined in (\ref{proofeq5}). Therefore, the following equation holds.
\begin{eqnarray}\label{proofeq6}
Z = T^{-1}\A_* T Y 
\end{eqnarray}
By multiplying $Y^{-1}$ from right hand side of (\ref{proofeq6}), $T^{-1} \A_* T$ is given by $T^{-1}A_* T = ZY^{-1}$. Since $Z$ is a cyclic matrix and $Y^{-1}$ is a block diagonal matrix, $T^{-1}A_* T$ is a cyclic matrix.  
\end{proof} 

While there is freedom for the coordinate transformations of each parameter matrix that is given as a cyclic reformulation, in addition to the obtained $T^{-1}$, this degree of freedom for the coordinate transformation can be achieved by using $\Phi^{-1}T^{-1}$ instead of $T^{-1}$, where $\Phi^{-1}$ is a block diagonal structure matrix given as follows: 
\begin{eqnarray}
\Phi = \left[\begin{array}{cccc}
          \Phi_1 & O_{n,n} & \cdots & O_{n,n}\\
          O_{n,n} & \ddots & \ddots & \vdots\\
          \vdots & \ddots  & \Phi_{M-1} & O_{n,n}\\
          O_{n,n} & \cdots & O_{n,n} & \Phi_M
     \end{array} \right].
\end{eqnarray}
Note that $\Phi_i (i=1,\cdots,M)$ should be regular matrices. 

Consequently, the cyclic identification algorithm of this paper is summarized as following Algorithm \ref{algo11}. 

\begin{algorithm}
\caption{System Identification for LPTV systems}
\label{algo11}
\begin{algorithmic}
\STATE [1.] Prepare cycled signals from the given input-output data. 
\STATE [2.] Compute $\A_*$, $\B_*$, $\C_*$,$\D_*$ using the existing subspace identification method with the cycled signals. 
\STATE [3.] Cyclic reformulation is derived using the obtained $\A_*$, $\B_*$, $\C_*$,$\D_*$ with the specific state coordinate transformation matrix $T$ from (\ref{coordinate}). 
\STATE [4.] Parameters in the time-varying state-space model $\A_k, \B_k, \C_k, \D_k$ are selected from the components of the cyclic reformulation $\check \A, \check \B, \check \C, \check \D$. 
\end{algorithmic}
\end{algorithm}

\section{Simulation}
In this section, numerical simulations of the proposed system identification algorithm are verified.

In the beginning, we verify whether Assumption \ref{ass1} is satisfied or not for the case that a subspace identification is applied for the cycled input and output. $n = 2$, $M = 3$, $m = l = 1$ is selected and the following linear time-varying system $P_{ex}$ is considered as a plant. 
\begin{eqnarray}
&A_0 = \begin{bmatrix}0&1\\0.5&1\end{bmatrix},A_1 = \begin{bmatrix}0&1\\0.9&-0.95\end{bmatrix},A_2 = \begin{bmatrix}0&1\\1&0.5\end{bmatrix}, \nonumber \\
&B_0 = \begin{bmatrix}1\\2\end{bmatrix}, B_1 = \begin{bmatrix}1.5\\2\end{bmatrix}, B_2 = \begin{bmatrix}1\\0.5\end{bmatrix}, \nonumber \\
&C_0 = \begin{bmatrix}1&0\end{bmatrix}, C_1 = \begin{bmatrix}1&0\end{bmatrix}, C_2 = \begin{bmatrix}1&0\end{bmatrix}, \nonumber \\
&D_0 = D_1 = D_2 = 0.5 \nonumber
\end{eqnarray}
The plant $P_{ex}$ is given as an observability companion form. Although $A_0, A_1$, and $A_2$ have unstable poles, $P_{ex}$ is stable in the meaning of a periodic system. The cyclic reformulation of $P_{ex}$ can be written as follows. 
\begin{equation}
     \check{A} = \left[\begin{array}{ccc}
          O_{2,2} & O_{2,2} & A_{2}\\
          A_0 &  O_{2,2} & O_{2,2}\\
          O_{2,2} & A_{1} &  O_{2,2} 
     \end{array} \right], \check{B} = \left[\begin{array}{ccc}
          O_{2,1} & O_{2,1} & B_{2}\\
          B_0 &  O_{2,1} & O_{2,1}\\
          O_{2,1} & B_{1} &  O_{2,1} 
     \end{array} \right], \nonumber
\end{equation}
\begin{equation}
     \check{C} = \left[\begin{array}{ccc}
          C_0 & O_{1,2} &O_{1,2} \\
          O_{1,2}  & C_1 & O_{1,2}\\
          O_{1,2} & O_{1,2}  & C_{2}
     \end{array} \right], \check{D} = \left[\begin{array}{ccc}
          D_0 & 0 & 0 \\
          0 & D_1 & 0\\
          0 & 0 & D_2
     \end{array} \right]. \nonumber
\end{equation}
where the matrix $\check S_1$ for $P_{ex}$ is given by
\begin{equation}
     \check S_1 = \left[\begin{array}{ccccc}
          0 & 1 & 0\\
          0 & 0 & 1\\
          1 & 0 & 0 
     \end{array} \right]. 
\end{equation}
By calculating (\ref{mar1}), the parameters $\check S_1^i H(i)$ for the plant $P_{ex}$ can be calculated as follows: 
\begin{eqnarray}
&\check H (0) = \left[\begin{array}{ccc}
          0.5 & 0 & 0 \\
          0 & 0.5 & 0\\
          0 & 0 & 0.5
     \end{array} \right], \check S_1 \check H (1) = \left[\begin{array}{ccc}
          1 & 0 & 0 \\
          0 & 1.5 & 0\\
          0 & 0 & 1
     \end{array} \right], \nonumber \\ &\check S_1^2 \check H (2) = \left[\begin{array}{ccc}
          2 & 0 & 0 \\
          0 & 2 & 0\\
          0 & 0 & 0.5
     \end{array} \right], \check S_1^3 \check H (3) = \left[\begin{array}{ccc}
          -1 & 0 & 0 \\
          0 & 2.5 & 0\\
          0 & 0 & 1
     \end{array} \right],\nonumber \\
     &\check S_1^4 \check H (4) =  \left[\begin{array}{ccc}
          1.5 & 0 & 0 \\
          0 & 3.5 & 0\\
          0 & 0 & -0.5
     \end{array} \right], \cdots \nonumber
\end{eqnarray}
We can see that $\check S_1^i \check H(i)$ is given as diagonal matrices as indicated in Lemma \ref{lemma1}. 

By applying an input $u(k)$ as shown in Fig. ~\ref{fig1}, we obtain an output $y(k)$ of the plant $P_{ex}$ as shown in Fig.~\ref{fig2}. Here, we assume $w(k) = 0$ and $v(k) = 0$ in this simulation. Note that the input $u(k)$ is randomly selected for each time step and is not an $M$-periodic signal. 

\begin{figure}[!h]
\centering
\includegraphics[width = 0.85\textwidth]{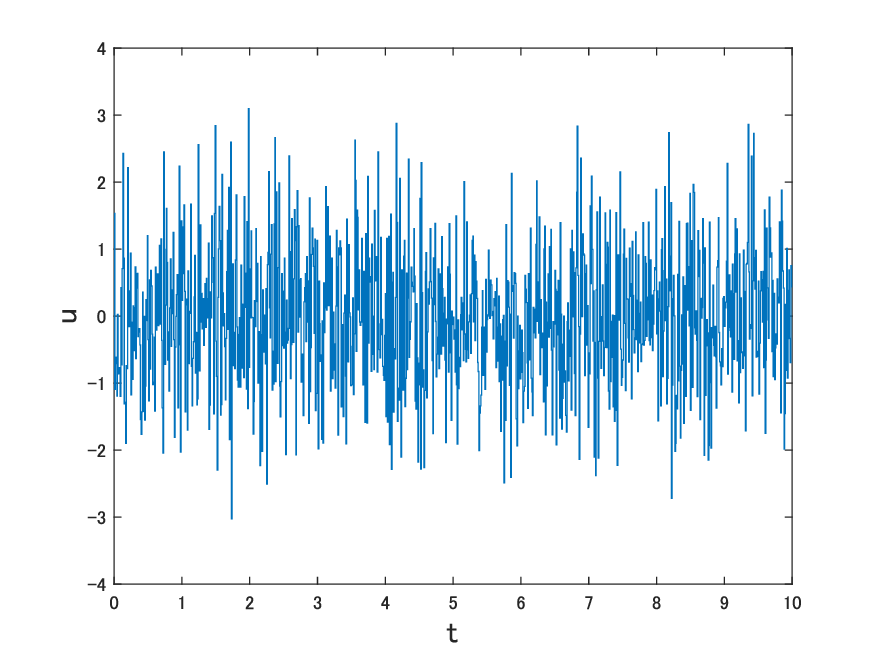}
\caption{Input for $P_{ex}$}
\label{fig1}
\includegraphics[width = 0.85\textwidth]{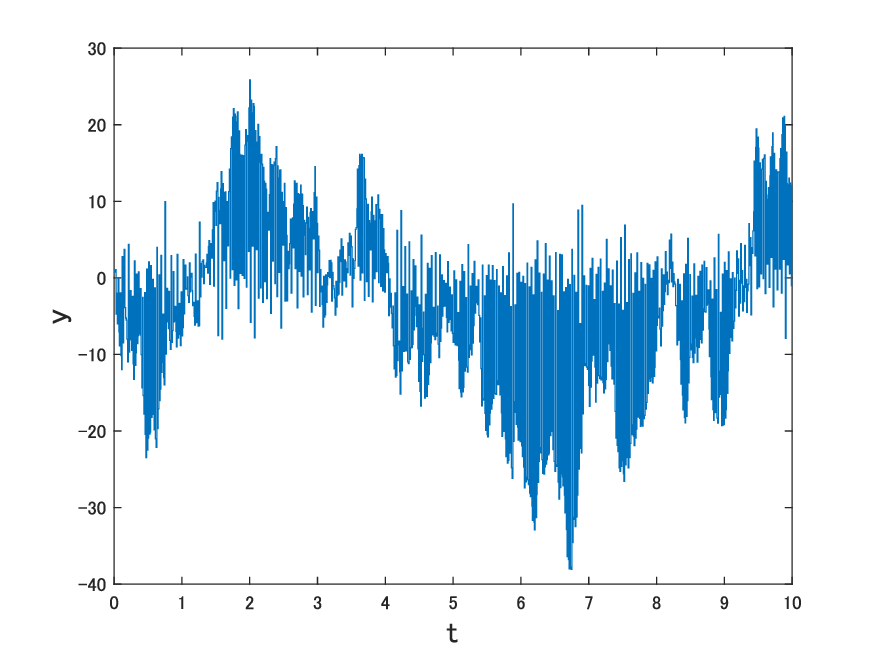}
\caption{Output of $P_{ex}$}
\label{fig2}
\end{figure}

Cycled signals $\check u(k)$ and $\check y(k)$ is obtained using Figs.~\ref{fig1} and \ref{fig2}. Then, a subspace identification method is applied for $\check u(k)$ and $\check y(k)$. We use the N4SID method, which is a kind of subspace identification method, in this simulation. The N4SID method is equipped with the "System Identification Toolbox" on MATLAB and can be easily implemented. 

The derived parameters $\A_*, \B_*, \C_*, \D_*$ by using the subspace identification are given as follows. 
\begin{eqnarray}
&\A_* = \left[\begin{array}{cccc}
   -0.234  &  0.316  &  0.167  &  0.102\\
   -0.427  & -0.397 &   0.427 &   0.315\\
    0.381  &  0.245 &  -0.094 &   0.655 \\
    0.534  & -0.220 &   0.341 &   0.318 \\
   -0.282  & -0.493 &  -0.490 &   0.473 \\
   -0.350  &  0.529 &   0.194 &   0.370 \end{array}\right.\nonumber \\
   & \left.\begin{array}{cc}
   0.720  &  0.253\\
   -0.109  & -0.348\\
    0.094  & -0.461\\
    -0.169  &  0.615\\
   0.153  &  0.214\\
   -0.172  &  0.255\end{array}\right]\label{astar01} 
 \end{eqnarray}
 \begin{eqnarray}
&\B_* = \begin{bmatrix}
   -0.007 &   0.024  & -0.006 \\
    0.018 &   0.016  & -0.000 \\
   -0.013 &   0.004  &  0.014 \\
    0.005 &   0.009  & -0.003 \\
   -0.016 &  -0.012  & -0.010 \\
    0.011 &  -0.027  & -0.000\end{bmatrix} \label{bstar01} 
    \end{eqnarray}
 \begin{eqnarray}
 &\C_* = \left[\begin{array}{cccc}
 -13.068  &  1.899 & 40.245 & -20.713\\ 
   30.059 & 175.521 & 100.751 &  35.170\\ 
  181.753 & -80.460  & 30.300 &  69.929 
\end{array}\right. \nonumber \\
& \left.\begin{array}{cc} -29.419  &  1.150 \\ 115.196 & 106.234 \\ -88.790 & 125.902 
\end{array}\right] \label{cstar01} \\
& \D_* = \begin{bmatrix}
    0.500 &  -0.000  & -0.000 \\
    -0.000 &   0.500  & 0.000 \\
    -0.000 &   0.000  &  0.500
\end{bmatrix}\label{dstar01} 
\end{eqnarray}
We calculate $\check S_1^i \check \Hm (i)$ for obtained $\A_*, \B_*, \C_*, \D_*$ as follows. 
\begin{eqnarray}
&\check \Hm (0) = \left[\begin{array}{ccc}
    0.500 &  -0.000  & -0.000 \\
    -0.000 &   0.500  & 0.000 \\
    -0.000 &   0.000  &  0.500
     \end{array} \right], \nonumber \\ &\check S_1 \check \Hm (1) = \left[\begin{array}{ccc}
          1.000 & 0.000 & 0.000 \\
          -0.000 & 1.500 & 0.000\\
          0.000 & 0.000 & 1.000
     \end{array} \right], \nonumber \\ &\check S_1^2 \check \Hm (2) = \left[\begin{array}{ccc}
          2.000 & 0.000 & 0.000 \\
          -0.000 & 2.000 & -0.000\\
          -0.000 & -0.000 & 0.500
     \end{array} \right], \nonumber \\ &\check S_1^3 \check \Hm (3) = \left[\begin{array}{ccc}
          -1.000 & -0.000 & -0.000 \\
          -0.000 & 2.500 & 0.000\\
          0.000 & 0.000 & 1.000
     \end{array} \right],\nonumber \\
     &\check S_1^4 \check \Hm (4) =  \left[\begin{array}{ccc}
          1.500 & 0.000 & 0.000 \\
          0.000 & 3.500 & 0.000\\
          -0.000 & -0.000 & -0.500
     \end{array} \right], \cdots \nonumber
\end{eqnarray}
Checking for each matrix $S^i \check \Hm (i)$, we can confirm that there are diagonal matrices for all $i$ in this simulation result. Although not shown here, it is confirmed that the diagonal matrix can be obtained in the same way when $i$ is $5$ or more. Therefore, we can confirm that Assumption \ref{ass1} holds. Furthermore, we can also confirm that $S^i\check \Hm (i)$ coincides with $S^i\check H(i)$ for $i=0,\cdots,4$. 

We give $F_j$ as follows: 
\begin{eqnarray}
F_1 = \begin{bmatrix}1\\0\\0\end{bmatrix}, F_2 = \begin{bmatrix}0\\1\\0\end{bmatrix}, F_3 = \begin{bmatrix}0\\0\\1\end{bmatrix}
  \end{eqnarray}
$\check F_j$ is given using (\ref{fj}) and their matrix size are $9 \times 3$. 
By applying step 3. in Algorithm \ref{algo11}, $\A_*, \B_*, \C_*$ and $\D_*$ are transformed by the following $T^{-1}$: 
 \begin{eqnarray}
&T^{-1} = \left[\begin{array}{cccc}
-13.07  &  1.90 &  40.25 & -20.71  \\
  -94.46 & -43.89 &  46.55 & 229.38    \\
   30.06 & 175.52 & 100.75 &  35.17  \\
   21.73 & 191.78 &  84.90 &  39.66   \\
  181.75 & -80.46 &  30.30 &  69.93  \\
   14.38 &  24.64 &   2.40 &   5.53   
\end{array}\right.\nonumber \\
   & \left.\begin{array}{cc}
   -29.42  &  1.15\\
   5.28 & -26.56\\
   115.20 & 106.23\\
   95.49 & 116.07\\
   -88.79 & 125.90\\
   -7.03 & -41.27
   \end{array}\right]. \label{henkanTinv}
   \end{eqnarray}
Since $T^{-1}$ is given as a regular matrix, the state coordinate transformation matrix $T$ is obtained as the inverse matrix of (\ref{henkanTinv}). Furthermore, the following matrices are obtained by applying a state coordinate transformation to (\ref{astar01}), (\ref{bstar01}), (\ref{cstar01}), and (\ref{dstar01}) using the state coordinate transformation matrix $T$.
\begin{eqnarray}
&\check \A = \left[\begin{array}{cccc}
  0.000  &  0.000&   0.000 &  -0.000   \\
    0.000 &   0.000 &   0.000 &  -0.000   \\
    0.000 &   1.000 &  -0.000 &  -0.000    \\
    0.500 &  1.000 &   0.000 &  -0.000   \\
   -0.000 &  -0.000 &  -0.000 &   1.000   \\
    0.000 &  -0.000 &   0.900 &  -0.950   \end{array}\right.\nonumber \\
   & \left.\begin{array}{cc}
    0.000  &  1.000\\
    1.000  &  0.500\\
         0 &  -0.000\\
     0.000 &   0.000\\
    0.000 &  -0.000\\
    -0.000 &  -0.000\end{array}\right]\label{astar02} 
 \end{eqnarray}
 \begin{eqnarray}
&\check \B = \begin{bmatrix}
    0.000 &  -0.000 &   1.000\\
    0.000 &  -0.000 &   0.500\\
    1.000 &  -0.000 &  -0.000\\
    2.000 &  -0.000 &  -0.000\\
   -0.000 &   1.500 &  -0.000\\
    0.000 &   2.000 &  -0.000\end{bmatrix} \label{bstar02} 
    \end{eqnarray}
\begin{eqnarray}
 &\check \C = \left[\begin{array}{cccc}    
 1.000 &  -0.000  & -0.000 &  -0.000    \\
   -0.000  &  0.000  &  1.000 &  -0.000 \\   
    0.000  &  0.000  &  0.000 &  -0.000    
\end{array}\right. \nonumber \\
& \left.\begin{array}{cc} 0.000 &  -0.000 \\ 0.000 &  -0.000 \\ 1.000      &   0.000
\end{array}\right] \label{cstar02} 
\end{eqnarray}
\begin{eqnarray}
& \check \D = \begin{bmatrix}
    0.500 &  -0.000  & -0.000 \\
    -0.000 &   0.500  & 0.000 \\
    -0.000 &   0.000  &  0.500
\end{bmatrix}\label{dstar02}
\end{eqnarray}
Thus, having obtained the cyclic reformulation, we can obtain $\A_k, \B_k, \C_k, \D_k$ as their elements. 
\begin{eqnarray}
&\A_0 = \begin{bmatrix}0.000&1.000\\0.500&1.000\end{bmatrix},\A_1 = \begin{bmatrix}-0.000&1.000\\0.900&-0.950\end{bmatrix},\nonumber \\
&\A_2 = \begin{bmatrix}0.000&1.000\\1.000&0.500\end{bmatrix}, \B_0 = \begin{bmatrix}1.000\\2.000\end{bmatrix},\nonumber \\ &\B_1 = \begin{bmatrix}1.500\\2.000\end{bmatrix}, \B_2 = \begin{bmatrix}1.000\\0.500\end{bmatrix}, \nonumber \\
&\C_0 = \begin{bmatrix}1.000&-0.000\end{bmatrix}, \C_1 = \begin{bmatrix}1.000&-0.000\end{bmatrix}, \nonumber \\ 
&\C_2 = \begin{bmatrix}1.000&0.000\end{bmatrix}, \nonumber \\
&\D_0 = \D_1 = \D_2 = 0.500 \nonumber
\end{eqnarray}
In this case, we can confirm that the obtained parameter matrices $\A_k, \B_k, \C_k, \D_k$ are well approximated with the parameter matrix of $P_{ex}$. Consequently, Problem \ref{prob1} is successfully solved for the case without noise. 

Next, we show the results of system identification using contaminated data. The input signal sequence is generated from a standard normal distribution and the process noise $w(k)$ is generated from the normal distribution whose mean and variance are zero and $1/5$, respectively. 
In this case, $\A_k, \B_k, \C_k, \D_k$ can be obtained by performing system identification using Algorithm \ref{algo11}. Note that Assumption \ref{ass1} is also satisfied for the case containing noise. The obtained model parameters are as follows. 
\begin{eqnarray}
&\check \A = \left[\begin{array}{cccc}
   -0.000 &   0.000 &   0.000 &  -0.000\\  
   -0.000 &   0.000 &  -0.000 &  -0.000  \\
    0.000 &   1.000 &   0.000 &  -0.000 \\
    0.505 &   1.000 &  -0.000 &  -0.000  \\
   -0.000 &  -0.000 &  -0.000 &   1.000  \\
    0.000 &   0.000 &   0.857 &  -0.908 
    \end{array}\right.\nonumber \\
   & \left.\begin{array}{cc}
    -0.000  &  1.000\\
    1.001  &  0.501\\
    -0.000 &  -0.000\\
  0.000 &   0.000\\
   0.000  & -0.000\\
    -0.000 &  -0.000 \end{array}\right]\nonumber 
 \end{eqnarray}
 \begin{eqnarray}
&\check \B = \begin{bmatrix}
   -0.000 &   0.000 &   1.011\\
   -0.000 &   0.000 &   0.501\\
    1.003 &  -0.000 &   0.000\\
    2.023 &   0.000 &   0.000\\
   -0.000 &   1.467 &  -0.000\\
   -0.000 &   1.978 &   0.000
    \end{bmatrix} \nonumber 
    \end{eqnarray}
 \begin{eqnarray}
 &\check \C = \left[\begin{array}{cccc}    
   1.000  &  0.000 & -0.000  &  0.000    \\
    0.000 &   0.000 &   1.000  & -0.000  \\  
   -0.000  & -0.000 &  -0.000  &  0.000    
\end{array}\right. \nonumber \\
& \left.\begin{array}{cc} 0.000 &  -0.000\\ 0.000 &  -0.000 \\ 1.000  &  0.000 
\end{array}\right] \nonumber 
\end{eqnarray}
\begin{eqnarray}
& \check \D = \begin{bmatrix}
    0.530  &  0.000 &   0.000\\
    0.000  &  0.497 &  -0.000\\
   -0.000 &   0.000 &   0.498
\end{bmatrix}
\end{eqnarray}
As a result, it can be confirmed that the obtained parameters are close to those of $P_{ex}$, respectively. In fact, its mean square error value is $0.0691$ and is sufficiently small despite the $20$ percent process noise. 
\section{Conclusion}

A cyclic reformulation based system identification algorithm for linear periodically time-varying plants has been proposed in this paper. First, properties of the Markov parameters of a cyclic reformulation of a periodic time-varying system are derived. The periodic time-varying parameters of a periodic time-varying system can be obtained by using the proposed state coordinate transformation matrix to transform the coordinates of the state-space model obtained using the cycled input and output signals. The effectiveness of the proposed system identification algorithm is verified using numerical examples. We verified that periodic time-varying systems can be identified with high accuracy without any special periodic inputs. 

Note: This work has been submitted to the IEEE for possible publication. Copyright may be transferred without notice, after which this version may no longer be accessible.


\begin{thebibliography}{99}

\bibitem{id00}
L. Ljung: System Identification — Theory for the User, 2nd Edition, PTR Prentice Hall (1999)

\bibitem{id01}
A. J. Helmicki, C. A. Jacobson, C. N. Nett: Control oriented system identification: a worst-case/deterministic approach in $H_\infty$, \emph{IEEE Transactions on Automatic Control}, Vol. 36, Issue 10, 1163-1176 (1991)

\bibitem{id-access1Robust}
Z. Li, L. Ma and Y. Wang: A New Robust Identification Algorithm for Hammerstein-Like System Using Identification Error Structure, \emph{IEEE Access}, Vol. 10, 29121 - 29131 (2022)

\bibitem{id02}
J. M. Bravo, T. Alamo and E. F. Camacho: Bounded Error Identification of Systems With Time-Varying Parameters, \emph{IEEE Transactions on Automatic Control}, Vol. 51, No. 7, 1144-1150 (2006)

\bibitem{id0}M. Verhaegen and X. Yu, A Class of Subspace Model Identification Algorithms to Identify Periodically and Arbitrarily Time-varying Systems, \emph{Automatica}, Vol. 31, No. 2, 201-216 (1995)

\bibitem{id03}
E. S. Tehrani and R. E. Kearney: A non-parametric approach for identification of parameter varying Hammerstein systems, \emph{IEEE Access}, Vol. 10, 6348–6362 (2022)

\bibitem{id1}
W. Yin and A. Saadat Mehr, Identification of linear periodically time-varying systems using periodic sequences, \emph{IEEE Control Applications \& Intelligent Control}, 1455-1459 (2009)

\bibitem{id2}
M. Yin, A. Iannelli and R. S. Smith,
Subspace identification of linear time-periodic systems with periodic inputs, \emph{IEEE Control Systems Letters}, Vol. 5, No. 1, 145-150 (2021)

\bibitem{id3}
J. Goos and R. Rintelon, Continuous-time identification of periodically parameter-varying state space models, \emph{Automatica}, Vol. 71, 254-263 (2016)

\bibitem{id4}
H. Oku: Recursive subspace model identification algorithms for slowly time-varying systems in closed loop, Proceedings of the European Control Conference 2007, 5715-5720 (2007)

\bibitem{id5}
T. Kawaguchi, M. Inoue and S. Adachi: State Estimation under Lebesgue Sampling and an Approach to Event-Triggered Control, \emph{SICE Journal of Control, Measurement, and System Integration}, Vol. 10, Issue 3, 259-265 (2017)

\bibitem{id6}
H. Tanaka and K. Ikeda: Identification of linear stochastic systems taking initial state into account, 56th Annual Conference of the Society of Instrument and Control Engineers, (2017)

\bibitem{id7}
Y. Fujimoto: Kernel Regularization in Frequency Domain: Encoding High-Frequency Decay Property, IEEE Control Systems Letters, Vol. 5, Issue 1, pp. 367-372 (2021)

\bibitem{id8}
Y. Fujimoto, I. Maruta and T. Sugie: Input Design for Kernel-Based System Identification From the Viewpoint of Frequency Response, IEEE Transactions on Automatic Control, Vol. 63, No. 9, pp. 3075-3082 (2018)

\bibitem{id9}
I. Maruta and T. Sugie: Closed-Loop Subspace Identification for Stable/ Unstable Systems Using Data Compression and Nuclear Norm Minimization, IEEE Access, Vol. 10, 21412-21423 (2022)

\bibitem{id10}
J. J. Gude, A. D. Teodoro, O. Camacho, and P. G. Bringas: A New Fractional Reduced-Order Model-Inspired System Identification Method for Dynamical Systems, \emph{IEEE Access}, Vol. 11, 103214 - 103231 (2023)

\bibitem{id11auto}
F. Felici, J.-W. van Wingerden and M. Verhaegen: Subspace identification of MIMO LPV systems using a periodic scheduling sequence, \emph{Automatica}, Vol. 43, No. 10, 1684 - 1697 (2007)

\bibitem{id11tac}
I. Uyanik, U. Saranli, M. M. Ankarali, N. J. Cowan and O. Morgul: Frequency-domain subspace identification of linear time-periodic (LTP) systems, \emph{IEEE Transactions on Automatic Control}, Vol. 64, No. 6, 2529 - 2536 (2019)

\bibitem{P3}
H. Kim, H. Shim, J. Back, and J. Seo, Consensus of multi-agent systems under periodic time-varying network. \emph{Proceedings of the 8th IFAC symposium on nonlinear control systems}, 155--160 (2010)

\bibitem{P2}
H. Okajima, Y. Hosoe, T. Hagiwara, State observer under multi-rate sensing environment and its design using $l_2$-induced norm, \emph{IEEE Access}, Vol.11, 20079-20087 (2023)

\bibitem{P1}
C. Scherer, Mixed $H_2$/$H_\infty$ control for time-varying and linear parametrically-varying systems, \emph{International Journal of Robust and Nonlinear Control}, Vol. 6, 929--952 (1996)

\bibitem{bit}
H. Okajima, Y. Hosoe, T. Hagiwara and Y. Minami: Basic Idea of Periodically Time-Varying Dynamic Quantizer in Networked Control Systems, \emph{2019 58th Annual Conference of the Society of Instrument and Control Engineers of Japan}, 883-889 (2019)

\bibitem{n4sid}P. V. Overschee and  B. D. Moor, N4SID: Subspace Algorithms for the Identification of Combined Deterministic-Stochastic Systems, \emph{Automatica}, Vol.30, No.1, 75--93 (1994)

\bibitem{real}
I. Markovsky, J. Goos, K. Usevich, R. Pintelon, Realization and identification of autonomous linear periodically time-varying systems, \emph{Automatica}, Vol. 50, No. 6, 1632--1640 (2014)

\bibitem{cyc1}
S. Bittanti and P. Colaneri, Invariant representations of discrete-time periodic systems, \emph{Automatica}, Vol. 36, Issue 12, 1777--1793 (2000)

\bibitem{cyc2}
S. Bittanti and P Colaneri, Periodic systems: filtering and control, \emph{Springer-Verlag, London} (2009)

\bibitem{cyc3}
P. Colaneri and V. Kucera, Model matching for periodic systems, \emph{Proc. of the American Control Conference}, 3143--3144 (1997)

\bibitem{cyc4}
Y. Hosoe, M. Miyamoto and T. Hagiwara, Cycling-based synthesis of robust output estimators for uncertain LPTV systems, \emph{SICE Journal of Control, Measurement, and System Integration}, Vol. 12, No. 2, 39--46 (2019)

\bibitem{cyc5}
H. Okajima, Design of Observer-Based Feedback Controller for Multi-Rate Systems With Various Sampling Periods Using Cyclic Reformulation, \emph{IEEE Access}, Vol. 11, 121956-121965 (2023)

\end{thebibliography}
\end{document}